\documentclass[article]{amsart}
\date{\today}

\usepackage{amsmath,amsthm}
\usepackage{amssymb}
\usepackage{bbm}
\usepackage{mathrsfs}               
\usepackage{enumerate}               
\usepackage{color}

\newtheorem{corollary}{Corollary}
\newtheorem{proposition}{Proposition}
\newtheorem{lemma}{Lemma}

\newtheorem{theorem}{Theorem}
\newtheorem{hypothesis}{Hypothesis}
\newtheorem{remark}{Remark}

\newcommand{\N}{\mathbb{N}} 
\newcommand{\C}{\mathbb{C}} 
\newcommand{\Z}{\mathbb{Z}} 
\newcommand{\R}{\mathbb{R}} 

\newcommand{\bx}{{\bf x}}

\newcommand{\SP}[2]{\bigg\langle #1,#2 \bigg\rangle} 
\newcommand{\Sps}[2]{\big\langle #1,#2 \big\rangle} 

\newcommand{\coretwo}{C_0^\infty((a, \infty), \C^2)}
\newcommand{\hilbert}{L^2(\R^2,\C^2)}

\newcommand{\ri}{\mathrm{i\,}}

\newcommand{\rd}{\mathrm{d}}

\newcommand{\curla}{{\rm curl\,}{\bf A}}
\renewcommand{\le}{\leqslant}

 \renewcommand{\ge}{\geqslant}
                                         
\title[Ballistic dynamics of Dirac particles in electro-magnetic
fields]{Ballistic dynamics of Dirac particles in electro-magnetic fields}
\author{Josef Mehringer}
\address{Josef Mehringer\\
Mathematisches Institut\\
Ludwig-Maximilians-Universit\"at\\
Theresienstra{\ss}e 39\\
D-80333 M\"unchen, Germany.}
\email{josef.mehringer@math.lmu.de}
\author{Edgardo Stockmeyer}
\address{ Edgardo Stockmeyer\\ Instituto de F\'\i sica\\
Pontificia Universidad Cat\'olica de Chile\\
Vicu\~na Mackenna 4860\\
 Santiago 7820436, Chile.}
\email{stock@fis.puc.cl}

\begin{document}
\begin{abstract}
  Investigating properties of two-dimensional Dirac operators coupled
  to an electric and a magnetic field (perpendicular to the plane)
  requires in general unbounded (vector-) potentials.  If the system
  has  a certain symmetry, the fields can be described by one-dimensional
  potentials $V$ and $A$.  Assuming that $|A|<|V|$
  outside some arbitrary large ball, we show that absolutely continuous
  states of the effective Dirac operators spread ballistically.
  These results are based on well-known methods in spectral dynamics
  together with certain new Hilbert-Schmidt bounds. We use Lorentz
  boosts to derive these new estimates.
\end{abstract}


\maketitle
\section{Introduction}\label{introduc}
It is well known that Dirac particles suffer from a phenomenon called
Klein tunneling. In dimension one, it can be roughly described as
follows : If one considers a step potential, for instance $V(x)=V_0$
for $x\ge0$ and zero otherwise, then massless Dirac particles coming
from the left will tunnel through the barrier independently of their
energy. As opposed to the classical quantum tunneling there is no
exponential damping factor diminishing the probability of finding the
particle on the right side of the barrier \cite{Klein1929,Thaller}. More
generally, one-dimensional massless Dirac particles spread as free
particles in the presence of electric fields. This effect has
attracted renewed attention due to the isolation of graphene in 2003
(see \cite{Novoselov2004}), since the low-energy charge carriers of
this material can be described by the two-dimensional massless Dirac
equation \cite{castro2009electronic,F2012,Fefferman2014}.  Indeed,
experiments have been carried out to observe Klein tunneling in
graphene confirming some theoretical predictions
\cite{PhysRevLett.98.236803,PhysRevLett.102.026807,young2009quantum}.

Consider the massless one-dimensional Dirac equation 
\begin{align*}
  -\ri \sigma_1\partial_1 +V 
\quad \mathrm{on}\quad L^2(\R,\C^2),
\end{align*}
with an electric potential $V\in L^1_{\rm loc}(\R)$, where $\sigma_1$
is the first Pauli matrix. In this case the Klein tunnel
effect is not very surprising from the mathematical point of view
since $ -\ri \sigma_1\partial_1 +V $ is unitarily equivalent to the free Dirac operator $-\ri
\sigma_1\partial_1$ by means of the transformation 
\begin{align}\label{intro1}
\exp
\left(\ri\sigma_1\int_0^xV(s)\mathrm{d}s\right).
\end{align}
However, in the
presence of magnetic fields the situation is different. In dimension two
 it is known that magnetic fields tend to localise Dirac
particles, very much like as in the Schr\"odinger case (see
\cite{Thaller}).

In a previous article we considered the combined electromagnetic
effect from a spectral theoretical point of view
\cite{MehringerStockmeyer2014}. In the present work we investigate
this further but focusing on the wave package spreading. Consider a
two-dimensional Dirac operator coupled to an electro-magnetic field
described by electric and magnetic potentials $V$ and ${\bf A}$. If
the field has translational or rotational symmetry the problem can be
reduced to the study of a family of Dirac operators on the line or on
the half-line, respectively. (Here the fields may be expressed through
one-dimensional potential functions $V$ and $A$.) Denote by $h$ one of
the members of these families. Our results roughly state the
following: Assuming that the function $\psi\not=0$ is of finite
energy and that it belongs to the absolutely continuous spectral
subspace of $h$ we obtain a lower bound on the C\`esaro mean of the
time evolution of the $p$-th moment ($p>0$), i.e. there is a
constant $C(\psi,p)>0$ such that
\begin{align}\label{d1}
  \frac{1}{T}\int_0^T \||x|^{p/2} e^{-\ri t h}\psi\|^2\,\rd t\ge C(\psi,p)T^p.
\end{align}
Besides certain regularity conditions the above inequality holds
provided $|{A}|<|V|$ outside some arbitrary large ball (see Theorems
\ref{lastmainthm1} and \ref{lastmainthm2}). As a consequence of the
causal behaviour of Dirac particles (see \cite[Theorem 8.5]{Thaller})
one has an upper bound of the same type, yielding altogether ballistic
dynamics. The consequences of inequalities of type \eqref{d1} for
two-dimensional Dirac operators with symmetries are summarised in
Corollaries \ref{appl1} and \ref{appl2}. We remark that if $V$ grows regularly at infinity the spectrum of $h$ is
purely absolutely continuous (see the discussion after Remark
\ref{barry}). The latter is in stark contrast to the behaviour of
non-relativistic particles. An important example is when the electric
and magnetic fields are asymptotically uniform, in which case $V$ and
$A$ grow linearly in the space coordinate.

The proof of the bounds of type \eqref{d1} are based upon the ideas of
\cite{Guarneri1989}, \cite{Combes1993} and \cite{Last1996}. These
results say roughly the following: Let $K\subset \R$ be a compact set
and $\mathbbm{1}_K$ be the characteristic function supported in $K$.
Then, the inequality \eqref{d1} holds if the function $\psi\in
\mathbbm{1}_{K}(h) L^2$ belongs to the absolutely continuous subspace
of $h$ provided a certain Hilbert-Schmidt bound is verified. This
latter condition demands the following for the product of
characteristic functions in space and energy: There is a constant
$C>0$ such that for all $I\subset \R$ compact
\begin{align}\label{intro2}
  \left\| \mathbbm{1}_{I}(x)  \mathbbm{1}_{K}(h)\right\|_{\rm{HS}} \le C_K\sqrt{|I|}.
\end{align} 
It is easy to check that the required bound is satisfied for the free
Dirac operator. For Schr\"odinger operator with potentials bounds like
\eqref{intro2} are obtained using semigroup properties combined with
perturbation theory \cite{Simon1982}.  However, in our case there is
no proper semi-group theory and, in addition, when the potentials are
allowed to grow at infinity, naive (resolvent) perturbation theory
gives estimates where the scaling in $|I|$ depends on the growth rate
of $A$ and $V$; that would eventually not deliver \eqref{d1}.  This is
not surprising since in this case $A$ and $V$ should not be treated as
perturbations to the free Hamiltonian. In the case $A=0$ one easily
sees that the transformation \eqref{intro1}  solves this problem. In
this work we provide new estimates of the type \eqref{intro2} for the
general case $V, A\not=0$ as long as $A$ is dominated by $V$ in
certain sense. Our approach is to use Lorentz boosts (of non-constant
speed) to transform the Hamiltonian to another operator with a
magnetic vector potential that vanishes at infinity. We remark that
the transformed operator is not going to be symmetric (see Section
\ref{loro}) since Lorentz boosts are not represented through unitary
maps in $L^2$ but only through invertible transformation (see
\cite[p. 70]{Thaller}). The relation between the original Hamiltonian
and the Lorentz transformed operator is made precise through certain
resolvent identities.
In the case of operators defined on the real line bounds like
\eqref{d1} are very much a corollary of \eqref{intro2} and
the proof of \cite[Theorem 6.2]{Last1996}. However, for Dirac operator defined on
the half-line one should proceed more carefully due to their
singularities at zero (c.f., Remark \ref{singatzero} and the
discussion at the beginning of Section \ref{last}).
\\
\\
\\
\noindent {\it This article is organised as follows: } In the next
section we state precisely our main results. The definition and basic
properties of the one-dimensional Dirac operators used here can be
found in Section \ref{basic}. In Section \ref{loro} we discuss the
behaviour of Dirac operators under certain Lorentz boosts of
non-constant speed and establish resolvent identities between original
and transformed operators.  We then apply the insight of Section
\ref{loro} to prove Theorems \ref{mainlemma} and \ref{hshk} in Section
\ref{proofhilbert}. The dynamical bounds for the half-line operators
(Theorem \ref{lastmainthm2}) are proven in Section \ref{last} where we
also establish a local compactness property suitable for our
regularity assumptions. In Appendix \ref{s.a.} we collect some
technical facts concerning self-adjointness.  We compute the resolvent
kernel for a half-line operator in Appendix~\ref{expres}. Finally, in
Appendix \ref{proofappl} we prove Corollaries \ref{appl1} and
\ref{appl2} about the consequences for two-dimensional Dirac operator
with symmetries.
\section{Main results and its applications}\label{mr}
The massless two-dimensional Dirac operator in an
electromagnetic field described by an electric potential $V$ and a
magnetic field $B$ (perpendicular to the plane) is given by
\begin{align}\label{hamiltonian}
H = {\boldsymbol \sigma}\cdot (-\ri \nabla- {\bf A}) +V  \quad 
\mathrm{on} \quad  {\mathcal H}:= \hilbert.
\end{align} 
Here ${\bf A}:\R^2\to\R^2$ is a magnetic vector potential satisfying
$B= \curla :=\partial_1 A_2-\partial_2 A_1$ and 
$\boldsymbol{\sigma} = (\sigma_1, \sigma_2)$ denote 
the Pauli matrices 
\begin{align*}
\sigma_1 = 
\begin{pmatrix}
0&1 \\
1 &0
\end{pmatrix}, \qquad
\sigma_2 = 
\begin{pmatrix}
0&-\ri \\
\ri &0
\end{pmatrix}.
\end{align*}
For a rotational symmetric magnetic field $B$ one can always choose
the rotational gauge
\begin{align}\label{rotationalgauge}
{\bf A} (\bx) =
\frac{1}{r^2}  \int_0^r B(s) s\,{\rm d} s  
\begin{pmatrix}
-x_2 \\
x_1
\end{pmatrix}
=: \frac{A(r)}{r}  
\begin{pmatrix}
-x_2 \\
x_1
\end{pmatrix},
\end{align}
where $r = |\bx|$. Thus, if $B$ and $V$ are
rotationally symmetric we can decompose $H$ as 
a direct sum operators defined on the half-line, i.e.
\begin{align}\label{mausi}
 H \cong \bigoplus_{k \in \, \Z+\frac{1}{2}}  h_k,
\end{align}
where 
\begin{align}\label{haensel}
h_k := - \ri \sigma_1 \partial_r  + \sigma_2 \left( \frac{k}{r}- A(r)\right) + V(r)
 \quad  \mathrm{on} \quad L^2(\R^+,\C^2).
\end{align}
Here in a slight abuse of notation we write $ V(|\bx|) = V(\bx)$.

On the other hand, if $B$  is translational symmetric, say,
in the $x_2$-direction, we can choose the Landau gauge  ${\bf
  A}(\bx)=(0,A(x_1))$, where
\begin{align}
  \label{landaugauge}
  A(x_1) = \int_0^{x_1}B(s){\rm d}s.
\end{align}
 Thus, if in addition $V$
has the same symmetry,  
the Hamiltonian $H$ can be represented as a direct integral
of one-dimensional fiber hamiltonians
\begin{align}\label{mausi2}
H \cong \int_\R^\oplus h(\xi) \, \mathrm{d}\xi,
\end{align}
with 
\begin{align}\label{gretel}
h(\xi) = -\ri \sigma_1\partial_1 + \sigma_2(\xi-A) +V 
\quad \mathrm{on}\quad L^2(\R,\C^2).
\end{align}
One-dimensional Dirac operators of type \eqref{haensel} and
\eqref{gretel} are the object of the next four theorems. 
Let us fix some notation: Throughout this article we denote by
$\mathbbm{1}_K$ the characteristic function on a set $K\subset \R$.
We write $P_{ac}(H)$ for the projection onto the absolutely continuous
subspace associated to a self-adjoint operator $H$. We will make use
of standard notation for norms: $\|\cdot\|_p$ denotes de $L^p$-norm,
$\|\cdot\|_T$ is the graph norm with respect to an operator $T$, and
$\|\cdot\|_{\rm HS}$ stands for the Hilbert-Schmidt norm.

In order to perform the afore mentioned Lorentz boosts (essentially of
velocity $A/V$) we introduce the following classes of electromagnetic
potentials:
\begin{hypothesis}[H1]
$A, V \in L^p_{\rm loc} (\R, \R)$ with $p \ge 2 $
such that $A = A_1 + A_2$, $V= V_1+V_2$, where $A_1,V_1$ have 
compact support and  $A_2, V_2 \in C^1(\R, \R)$ fulfill
\begin{enumerate}[{\rm i)}]
\item $V_2$ is supported away from $0$ and
          $ {\rm supp}(A_2) \subset  {\rm supp}(V_2)$,
\item $\|A_2/V_2\|_\infty <1$,
\item the derivative $(A_2/V_2)'$ is bounded on $\R$.
\end{enumerate}
\end{hypothesis}
\begin{hypothesis}[H2]
$A, V \in L^p_{\rm loc} ([0,\infty), \R)$ with $p>2$ 
such that $A = A_1 + A_2$, $V= V_1+V_2$, where $A_1,V_1$ have 
compact support and  $A_2, V_2 \in C^1(\R^+, \R)$ fulfill
\begin{enumerate}[{\rm i)}]
\item $ V_2$ are supported away from $0$ and
          $ {\rm supp}(A_2) \subset  {\rm supp}(V_2)$,
\item $\|A_2/V_2\|_\infty <1$,
\item the derivative $(A_2/V_2)'$ is bounded on $[0, \infty)$.
\end{enumerate}
\end{hypothesis}
\begin{theorem}\label{mainlemma}
For $\xi \in \R$ let $h(\xi)$  be given as in \eqref{gretel} with $A,V$
satisfying Hypothesis {\rm (H1)}.
Then there is a constant $C_\xi>0$ such that for any compact
intervall $I\subset\R$ we have 
\begin{align*}
  \left\| \mathbbm{1}_{I}(h(\xi)-\ri)^{-1} \right\|_{\rm{HS}} \le C_\xi\sqrt{|I|}.
\end{align*}
\end{theorem}
\noindent
A direct consequence of this HS-bound is the following theorem whose
proof is  the same as the one of Theorem 6.2 of \cite{Last1996}. 
\begin{theorem}\label{lastmainthm1}
Consider the operator $h(\xi)$, $\xi\in \R$, with $A,V$ satisfying
Hypothesis {\rm (H1)}.
Let  $\Delta\subset \R$ be a bounded energy interval and 
$\psi \in P_{ac}(h(\xi))\mathbbm{1}_{\Delta}(h(\xi)) L^2(\R, \C^2)$ be
non-zero. Then, for each $p>0$, there is a constant 
$C_\xi (\psi, \Delta ,p)$ such that
\begin{align}\label{theineq1}
\langle  \|x^{p/2} e^{-\ri t h(\xi)} \psi\|^2\rangle_T
\ge C_\xi(\psi, \Delta ,p) T^{p}
  \end{align}
for all $T>0$.
\end{theorem}
\noindent
In the case of the half-line operators $h_k$ we obtain similar 
results:
\begin{theorem}\label{hshk}
For $k \in \Z+\tfrac{1}{2}$ let $h_k$ be given as in \eqref{haensel}, 
with $A,V$ satisfying Hypothesis {\rm (H2)}. Then there is a
constant $C_k>0$ such that for any  compact intervall $I\subset
[1,\infty)$ we have
 \begin{align*}
  \left\| \mathbbm{1}_{I}(h_k-\ri)^{-1} \right\|_{\rm{HS}} \le C_k\sqrt{|I|}.
\end{align*}
\end{theorem}
\begin{remark}\label{singatzero}
  Note that the operators $h_k$ have a $k/x$-singularity at zero. We
  do not need boundary conditions to define them.  To deduce the
  HS-bounds for $h_k$ we compare them with an auxiliar operator, which
  is regular at zero, and satisfies certain boundary conditions.
  Because of that, we obtain  the HS-bounds only for intervalls
  $I\subset \R$
  supported  away from zero.
\end{remark}
\noindent
A consequence of the latter result is the following:
\begin{theorem}\label{lastmainthm2}
Consider the operator $h_k$, $k\in \Z+\tfrac{1}{2}$
with $A,V$ satisfying Hypothesis {\rm (H2)}. 
Let $\Delta\subset \R$ be a bounded energy intervall and $\psi
\in P_{ac}(h_k)\mathbbm{1}_{\Delta}(h_k) L^2(\R^+, \C^2)$ be
non-zero. Then, for each $p>0$, there is a constant 
$C_k (\psi,\Delta, p)$ such that 
\begin{align}\label{theineq}
\langle  \|x^{p/2} e^{-\ri t h_k} \psi\|^2\rangle_T 
\ge  C_k (\psi,\Delta, p)\, T^{p}  
  \end{align}
for all $T>0$.
\end{theorem}
\noindent
This statement is proven in Section \ref{last}.  We have to modify the
argument of \cite{Last1996} since Theorem \ref{hshk} is only  valid for
intervals with non-vanishing distance to zero.  As an additional ingredient we use the local
compactness of $h_k$ (also proven in Section \ref{last}) for $A,V \in
L_{\rm loc}^p([0,\infty),\R)$ with $p>2$.
\begin{remark}\label{barry}
  We note that theorems \ref{mainlemma} - \ref{lastmainthm2} also hold
  for massive Dirac operators, i.e.  operators of the form $h(\xi) +
  m\sigma_3$ or $h_k + m\sigma_3$ with a constant $m$. Here
$\sigma_3 = -\ri \sigma_1\sigma_2$ is the third Pauli matrix.
\end{remark}
Since Theorems \ref{lastmainthm1} and \ref{lastmainthm2} apply only if
we have some absolutely continuous spectrum we mention some
interesting examples: 
Let $A,V$ satisfy {\rm (H1)} and,  in addition, 
\begin{itemize}
\item $(A_2/V_2)'$ is integrable at $\pm \infty$,
\item $|V_2 (x)| \to \infty$ as $|x| \to \infty$,
\end{itemize}
then $h(\xi)$, $\xi\in\R$, has purely absolutely continuous spectrum with
$\sigma_{\rm ac}(h(\xi)) =\R$.  Similarly, if the
potentials $A,V$ satisfy {\rm  (H2)} and 
\begin{itemize}
\item $(A_2/V_2)'$ is integrable at $\infty$,
\item $|V_2 (x)| \to \infty$ as $x \to \infty$,
\end{itemize}
then $h_k$, $k \in \Z+\frac{1}{2}$, has purely absolutely continuous spectrum
with $\sigma_{\rm ac}(h_k) =\R$ (see \cite{KMSYamada}, Propositions 1 and
2) .

Finally we illustrate how one can harness Theorems \ref{lastmainthm1}
and \ref{lastmainthm2} to Dirac operators in dimension two. 
\begin{corollary}\label{appl1}
  Let $H$ be given as in \eqref{hamiltonian} with translational
  symmetric $B$ and $V$. Let $A(x):=\int_0^x B(s)ds$ and $V$ satisfy {\rm (H1)}.  Let $\Delta
  \subset \R$ be bounded and $p>0$. Then for any $\psi \in \mathbbm{1}_\Delta(H)
  \mathcal{H}$ such that the set
\begin{align}\label{condappl1} 
\big\{\xi \in \R \,|\, 
\widehat\psi (\, \cdot \, , \xi) \neq 0, \
\widehat\psi (\, \cdot \, , \xi) \in 
P_{ac}(h(\xi)) L^2(\R,\C^2)\big\},
\end{align}
has non-trivial Lebesgue measure, there exist a constant $C(\psi,
\Delta, p)>0$ such that
\begin{align*}
\langle\| |x_1|^{p/2}e^{-\ri tH } \psi\|^2\rangle_T 
\ge C(\psi, \Delta, p) \,T^{p}
\end{align*}
for all $T>0$. Here $\widehat \psi $ denotes the Fourier-transform 
of $\psi$ in the $x_2$-variable, i.e. we use the notation
$\widehat \psi (x_1, \,\cdot\,) = \mathcal{F}_{x_2} \psi(x_1,\,\cdot\,)$.
\end{corollary}
\begin{corollary}\label{appl2}
Let $H$ be given as in \eqref{hamiltonian} with rotational symmetric $B$
and  $V$. Let $A(x):=x^{-1}\int_0^x B(s)s \rd s$ and $V$
satisfy {\rm  (H2)}. Let $\Delta
  \subset \R$ be bounded and $p>0$. Then for 
$\psi \in  P_{ac}(H)\mathcal{H} \cap \mathbbm{1}_\Delta(H) \mathcal{H}$,
with $\psi \neq 0$,  there exist a constant 
$C(\psi, \Delta, p) >0$ such that
\begin{align*}
\langle\| |\bx|^{p/2} e^{-\ri tH}\psi \|^2\rangle_T \ge 
C(\psi, \Delta, p) \,T^{p}
\end{align*}
for all $T>0$.
\end{corollary}
\section{Basic properties of Dirac operators in dimension one}\label{basic}
In this article we  basically work with two different types of one-dimensional
Dirac operators. The first type is given by
\begin{align}\label{charlotte}
h = \sigma_1 (-\ri \partial_x) - \sigma_2 A + V
\quad \mbox{on} \quad L^2(\R, \C^2), 
\end{align}
with $A,V \in L^2_{\rm{loc}}(\R, \R)$. We note that $h$ is in the
limit point case at $\pm \infty$ \cite[Korollar
15.21]{Weidmann2}. Then according to \cite[Chapter
15]{Weidmann2} (see also \cite{Weidmann0})
the operator is essentially self-adjoint on
\begin{equation}\label{do}
\mathcal{D}_0(h)=\left\{
\psi\in \mathcal{D}_{\rm max}(h)\, |\,
\psi \,\mbox{has compact support in}\, \R 
\right\},
\end{equation}
where
\begin{align}\label{lara}
\mathcal{D}_\mathrm{max} (h) = 
\left\{
\psi \in L^2(\R, \C^2)\, |\, \psi \, \mbox{abs. cont.,} \,
h\psi \in L^2(\R, \C^2)
\right\}.
\end{align}
In fact by Lemma \ref{good-core} in
Appendix \ref{s.a.} we know that $h$ is essentially self-adjoint on
$C_0^\infty (\R,\C^2)$. We denote its self-adjoint extension by $h$ again. 

The second type of operators is defined on the half-line
\begin{align}\label{amelie}
h_k = \sigma_1 (-\ri \partial_x) + \sigma_2\left(\tfrac{k}{x} -A\right) + V
\quad \mbox{on} \quad L^2((0, \infty), \C^2), 
\end{align}
with $A,V \in L^2_{\rm{loc}}([0,\infty), \R)$ and $k \in (\Z+
\tfrac{1}{2} )\cup \{0\}$. The maximal domain of $h_k$ is given by 
\begin{align}\label{lara2}
\mathcal{D}_\mathrm{max} (h_k) = 
\{\psi \in L^2(\R^+, \C^2)\, |\, \psi \ \mathrm{abs. cont.}, \
h_k\psi \in L^2(\R^+, \C^2))\}.
\end{align}
For these half-line operators
we distinguish two cases: When $|k|\ge \tfrac{1}{2}$ the operator
$h_k$ is in the limit point case at $+\infty$ \cite[Korollar
15.21]{Weidmann2} and in the limit point case at $0$ (see Proposition
\ref{lpcin0} in Appendix \ref{s.a.}). Thus $h_k$ is essentially
self-adjoint on $\mathcal{D}_0(h_k)$ (defined as in \eqref{do}). We
denote its self-adjoint extension by $h_k$ again.  Using  Lemma
\ref{good-core} in Appendix \ref{s.a.} it actually holds that
$C_0^\infty((0,\infty),\C^2)$ is also an operator core for $h_k$. 
\begin{remark}\label{ws1}
 We note that by \cite[Satz 15.6]{Weidmann2} the domains of
self-adjointness of $h$ and $h_k$, $|k|\ge 1/2$, coincide with their
maximal domains.
\end{remark}
In the case when $k=0$ the operator is in the limit
point case at $+\infty$ and in the limit circle case at $0$. According
to the theory of Sturm-Liouville operators (see \cite[Satz
15.12]{Weidmann2} or \cite{Weidmann0}) $h_0$ has a one-parameter
family of self-adjoint
realisations with corresponding domains 
\begin{equation}\label{domain1}
\mathcal{D}^\alpha(h_0) =
\big\{ \mathcal{D}_\mathrm{max} (h_0)\, |\, 
\lim_{x \to 0}\psi_1(x)\cos\alpha -\psi_2(x)\sin \alpha =0\big
\},\quad \alpha \in [0, 2\pi). 
\end{equation}
In the sequel we work with the self-adjoint realisation of $h_0$ on
$\mathcal{D}^0(h_0)\equiv \mathcal{D}(h_0)$ and denote, as before, the resulting operator by the same
symbol.  
\begin{remark}\label{ws2}
  Let $\chi\in C^\infty((0,\infty),[0,1])$ be a smooth function
  supported away from zero with bounded first derivative. By the
  definition of the domains of self-adjointness \eqref{lara2} (see Remark
  \ref{ws1}) and \eqref{domain1} and the fact that $k/x $ is a bounded
  function on the support of $\chi$ we have that
  \begin{align}\label{chid}
    \chi \mathcal{D}(h_k)\subset\mathcal{D}(h_0),
  \end{align}
whenever the two operators have the same potentials $V$ and $A$.
\end{remark}
\begin{proposition}\label{basicHS}
Consider $h, h_0$ with $V\in L^2_{\rm loc}$ and $A=0$. Then for bounded
intervals $I \subset \R$, $I_0\subset (0, \infty)$ 
the operators $\chi_I (h-\ri)^{-1}$, $\chi_{I_0} (h_0-\ri)^{-1}$
are Hilbert-Schmidt with Hilbert-Schmidt norms 
\begin{equation}\label{HSbound1}
\left\| \chi_I \frac{1}{(h-\ri)} \right\|_{\rm HS}
\le \frac{1}{\sqrt{2}} \,|I|^{1/2}\,,
\end{equation}
\begin{equation}\label{HSbound2}
\left\| \chi_{I_0} \frac{1}{(h_0-\ri)} \right\|_{\rm HS}
\le |I_0|^{1/2}\,.
\end{equation}
\end{proposition}
\begin{proof}
Intertwining the resolvents by the unitary transformation 
\begin{align}\label{theu}
\big[U\psi] (x) = 
\exp \left(\ri\sigma_1\int_0^xV(s)\mathrm{d}s\right)\psi(x)
\end{align}
on $L^2(\R, \C^2)$, respectively on $L^2((0, \infty), \C^2)$, the
proof reduces to the case $V=0$. Then, the statement on
$\chi_I (h-\ri)^{-1}$ is a direct consequence of the Kato-Seiler-Simon 
inequality (see \cite{SeilerSimon} or \cite[Thm. 4.1]{Simon_TI}). For the claim on the
operator $\chi_{I_0} (h_0-\ri)^{-1}$ on $L^2((0, \infty), \C^2)$
we use directly the resolvent kernel (computed in the Appendix~\ref{expres}) given by
\begin{equation}\label{resolventkernel}
\frac{1}{\sigma_1(-\ri \partial_x) -\ri}(x_1,x_2) =
\begin{cases}
ie^{-x_1} 
\begin{pmatrix}
\sinh x_2 & \cosh x_2 \\
\sinh x_2 & \cosh x_2
\end{pmatrix}
& \mathrm{for} \ x_1 > x_2 \ge 0 \\[0.5cm]
ie^{-x_2}
\begin{pmatrix}
\sinh x_1 & -\sinh x_1 \\
-\cosh x_1 & \cosh x_1
\end{pmatrix}
& \mathrm{for} \ x_2 > x_1 \ge 0.
\end{cases}
\end{equation}
Since $\chi_{I_0}(x_1) (\sigma_1(-\ri \partial_x)-\ri)^{-1}(x_1,x_2)$ is
square-integrable, we obtain that $\chi_{I_0} (h_0-\ri)^{-1}$ is a
Hilbert-Schmidt operator. In fact this is true for any function
$f\in L^2((0,\infty))$; the norm  may be computed as 
\begin{equation}\label{hswithf}
\begin{split}
\left\| f \frac{1}{(h_0-\ri)} \right\|_{\rm{HS}}^2 &=
\int_0^\infty \int_0^\infty 
|f(x_1)|^2
\left\|\frac{1}{\sigma_1(-\ri \partial_x)-\ri}(x_1,x_2)\right\|^2_{M_2(\C)}
\mathrm{d}x_1\mathrm{d}x_2 \\ 
&=
\int_0^\infty 
|f(x_1)|^2 \int_0^\infty  e^{-2|x_1-x_2|} \, \mathrm{d}x_2\mathrm{d}x_1 
\le \|f\|_2^2.
\end{split}
\end{equation}
\end{proof}
\section{Lorentz transformations and resolvent identities}\label{loro}
It is known from the classical theory of electrodynamics that Lorentz
boosts enables one to transform magnetic fields into electric ones and
vice-versa.  We are allowed to use this principle here, since the time
dependent Dirac equation is invariant under Lorentz transformations
(see Chapter 3 and Section 4.2 of \cite{Thaller}).  In this section we
use Lorentz boosts to transform the Hamiltonian to another operator
whose magnetic vector potential vanishes at infinity. 
 In our case the speed of the boost is given by
the ratio of $A$ and $V$ and will, therefore, depend on the space
variable.  Recall that in $2+1$ space-time dimension a Lorentz boost in direction
${\bf n}\in \R^2$ with speed $\beta<1$ (in general, smaller than the speed of
light) is represented by the operator $L_{\Lambda}=e^{{\bf
    n}\cdot{\boldsymbol \sigma}\theta/2}$, where
$\beta=\tanh{\theta}$.

Let us  point out a transformation property 
of the $\sigma$-matrices: 
Let $a=0$ or $a = -\infty$ and 
$\theta \in C^1((a, \infty), \R)$. Observe that
\begin{align*}
e^{-\sigma_2\theta /2 } \sigma_1 e^{\sigma_2\theta/2 } 
& = e^{-\sigma_2\theta  } \sigma_1 =  (\cosh \theta  -   \sigma_2 \sinh \theta) \sigma_1,
\end{align*} 
therefore, 
\begin{align}\label{lorentzfree}
\begin{split}
e^{-\sigma_2\theta/2 }  \sigma_1 (-\ri \partial_1) e^{\sigma_2\theta/2 } 
& = e^{-\sigma_2\theta/2 } \sigma_1  e^{\sigma_2\theta/2 }  
e^{-\sigma_2\theta/2 } (-\ri \partial_1) e^{\sigma_2\theta/2 } \\
& = (\cosh \theta  -   \sigma_2 \sinh \theta) \sigma_1
\big(-\ri \partial_1 -\ri \sigma_2 \tfrac{\theta'}{2}\big) \\
& = \cosh \theta (1 -   \sigma_2 \tanh \theta) \sigma_1
\big(-\ri \partial_1 -\ri \sigma_2 \tfrac{\theta'}{2}\big)
\end{split}
\end{align} 
on the subspace $\coretwo$. Recall that for potentials $V,A$
satisfying Hypothesis (H1) (for $a=-\infty$) or (H2) (for $a=0$) the
ratio fulfills
$\|A_2/V_2\|_\infty<1$. This enable us to define the following objects: Let $\theta(x) = \tanh^{-1}(\beta(x))$,
where $\beta:= A_2/V=A_2/V_2$ and  $\gamma:=\cosh \theta$ (then clearly,
$\gamma^{-1} =\sqrt{(1-\beta^2)}$).  For the Dirac operator $h$
on $L^2(\R,\C^2)$ (as given  in \eqref{charlotte}), with potentials
satisfying (H1), we compute
\begin{align}\label{maintransformA} 
\begin{split}
e^{-\sigma_2\theta/2 } h e^{\sigma_2\theta/2 }  
& = 
\gamma (1 -   \sigma_2 \beta) \sigma_1
\big(-\ri \partial_1 -\ri \sigma_2 \tfrac{\theta'}{2}\big) 
+ \big(1-\sigma_2 \tfrac{A_2}{V}\big) V - \sigma_2 A_1\\
&=
M
\Big[\sigma_1(-\ri \partial_1) +V/\gamma 
-\gamma(1 + \sigma_2\beta)\sigma_2 A_1 
+ \sigma_3\tfrac{\theta'}{2} \Big]
\end{split}
\end{align}
on $C_0^\infty(\R,\C^2)$, where $M:= \gamma(1-\sigma_2\beta)$ is a
bounded multiplication operator with bounded inverse.

Analogously, for the operator $h_k$ (defined in
\eqref{amelie}), with potentials satisfying (H2), we have
\begin{align}\label{maintransformB}
\begin{split}
e^{-\sigma_2\theta/2 } &h_k e^{\sigma_2\theta/2 }  =
M
\Big[\sigma_1(-\ri \partial_1) +V/\gamma 
+\gamma(1 + \sigma_2\beta)\sigma_2 \big(\tfrac{k}{x} -A_1\big)
+ \sigma_3\tfrac{\theta'}{2} \Big]
\end{split}
\end{align}
on $C_0^\infty((0,\infty),\C^2)$. Note that, abusing notation, we use the
 symbol $M$ in both cases, however, the former acts on $L^2(\R^2,\C^2)$
and the latter on $L^2((0,\infty),\C^2)$.  Summarising, we have the following
identities on $C_0^\infty(\R,\C^2)$ and
$C_0^\infty((0,\infty),\C^2)$,
\begin{equation}\label{them}
e^{-\sigma_2\theta/2 } h e^{\sigma_2\theta/2 }=M\tilde{h}\quad\mbox{and}\quad
e^{-\sigma_2\theta/2 } h_k e^{\sigma_2\theta/2 }=M\tilde{h}_k,
\end{equation}
respectively. Where the operator
\begin{equation*}
\tilde h := 
\sigma_1(-\ri \partial_1) 
-\gamma(1 + \sigma_2\beta)\sigma_2 A_1
+V/\gamma + \sigma_3 \tfrac{\theta'}{2}
\quad \mathrm{on} \ L^2(\R,\C^2),
\end{equation*}
is essentially self-adjoint on $C_0^\infty(\R, \C^2)$. Similarly,
\begin{equation*}
\tilde h_k := 
\sigma_1(-\ri \partial_1) 
+\gamma(1 + \sigma_2\beta)\sigma_2 \big(\tfrac{k}{x} -A_1\big)
+V/\gamma + \sigma_3 \tfrac{\theta'}{2}
\quad \mathrm{on} \ L^2((0,\infty),\C^2),
\end{equation*}
is, for $k \in \Z + \tfrac{1}{2}$, in the limit point case at $0$ and
thus essentially self-adjoint on $C_0^\infty((0,\infty),\C^2)$. The
latter holds since $A_2=0$ in a vicinity of zero (and hence so is
$\beta$). Therefore, $M\tilde h $ and $M\tilde h_k $ are closed
operators on the domains $\mathcal{D} (\tilde h)$ and  $\mathcal{D}
(\tilde h_k)$, respectively.
\begin{lemma}\label{resolventestimate}
Consider the operator $h$ and $h_k$ with potentials satisfying
Hypotheses {\rm (H1)} and {\rm (H2)}, respectively. Then 
$\mathcal{D} (\tilde h) = \mathcal{D} (M\tilde h) 
=e^{-\sigma_2\theta/2 }\mathcal{D} (h)$, respectively
$\mathcal{D} (\tilde h_k) = \mathcal{D} (M\tilde h_k) 
=e^{- \sigma_2\theta/2}\mathcal{D} (h_k)$ 
for $k \in \Z + \tfrac{1}{2}$. In addition, the resolvent sets fulfill
$ \varrho( h)= \varrho(M\tilde h)$ and for any $z\in  \varrho(h)$ we have
\begin{align*}
&(M\tilde h-z)^{-1}=e^{- \sigma_2\theta/2} (h-z)^{-1}e^{
  \sigma_2\theta/2},\\&\big\| (M\tilde h-z)^{-1}\big\| \le 
\big\|( h-z)^{-1}\big\| \,\big\| e^{|\theta|} \big\|_{\infty}.
\end{align*}
Similarly, for $k \in \Z + \tfrac{1}{2}$,  $ \varrho( h_k)= \varrho(M\tilde h_k)$ and for any $z\in
\varrho(h_k)$ holds
\begin{align*}
 &(M\tilde h_k-z)^{-1}=e^{- \sigma_2\theta/2} (h_k-z)^{-1}e^{
  \sigma_2\theta/2},\\&\big\| (M\tilde h_k-z)^{-1}\big\| \le 
 \big\|( h_k-z)^{-1}\big\|\,\big\| e^{|\theta|} \big\|_{\infty}.
\end{align*}
\end{lemma}
\begin{proof}
  We give the proof only for the operator  $h$, since for $h_k$
  one can proceed analogously.  The equality
  $\mathcal{D}(\tilde h)=\mathcal{D}(M\tilde h)$ is a direct
  consequence of the bounded invertibility of $M$ (with inverse
  $M^{-1} = \gamma(1+\sigma_2\beta)$). Note that the relations
  \eqref{them} are also valid for $C^1$-functions with compact support,
  which form an invariant space under 
  $e^{\pm\theta/2 \sigma_2}$ transformations.
  In addition, since $C^1_0(\R,\C^2)$ is contained in 
  $\mathcal{D} (\tilde h)$ and in $\mathcal{D} ( h)$, it is also an
  operator core for $\tilde h$ and $h$. Using \eqref{them} we easily
  get that, for any $f \in C^1_0(\R,\C^2)$,
  \begin{align}
    \label{eq:1}
    \big\|e^{-\sigma_2\theta/2}f \big\|_{M\tilde h} & \le
    \|e^{-\sigma_2\theta/2}\|_\infty \|f\|_{h},\\\label{eq:2}
    \big\|e^{\sigma_2\theta/2}f \big\|_{ h} &\le
    \|e^{\sigma_2\theta/2}\|_\infty \|f\|_{M \tilde h}.
  \end{align}
Let $\varphi\in \mathcal{D}(h)$ and  $(\varphi_n)_{n\in \N}\subset
C^1_0(\R,\C^2)$ be a sequence that converges to $\varphi$ in the
$h$-graph norm. Due to \eqref{eq:1} the sequence 
$(e^{-\sigma_2\theta/2}\varphi_n)_{n\in\N}$
is  Cauchy in the $M\tilde h$-graph norm. Hence 
$$\lim_{n\to\infty}e^{-\sigma_2\theta/2}\varphi_n=e^{-\sigma_2\theta/2}\varphi\in 
\mathcal{D}(M \tilde h).$$ Thus we get that
$e^{-\sigma_2\theta/2}\mathcal{D}(h)\subset\mathcal{D}(M \tilde
h)$. The opposite inclusion can be shown along the same lines using
the inequality \eqref{eq:2}.
 
In order to derive the resolvent bound observe that 
$\mathcal{D}(M\tilde h)=
e^{-\sigma_2\theta/2 }\mathcal{D}(h)$ implies the operator identity,
for any $z\in \C$, 
\begin{align*}
e^{-\sigma_2\theta/2 } (h-z) e^{\sigma_2\theta/2 }
= (M \tilde h-z) \quad 
\mathrm{on} \quad \mathcal{D}(M \tilde h).
\end{align*}

Let $z\in \varrho(h)$. Since $h-z: \mathcal{D}(h) \to L^2(\R,\C^2)$ is bijective
with bounded inverse we conclude 
$M\tilde h-z : \mathcal{D}(M\tilde h) \to L^2(\R,\C^2)$ 
is also bijective with bounded inverse. In addition,
\begin{align*}
\big\| 
(M\tilde h-z)^{-1} \big\| \le
\big\| e^{\theta/2 \sigma_2} \big\|
\big\|(h-z)^{-1} \big\|
\big\| e^{-\theta/2 \sigma_2} \big\|.
\end{align*}
\end{proof}
We close this section by illustrating how to deduce resolvent
identities of the type presented  in Lemma \ref{resolventestimate}
when $V<A$ at $\infty$. To this end we consider the operator $h$
as given in \eqref{charlotte} with $A,V$ fulfilling
\begin{hypothesis}[H1$^\prime$]
$A, V \in L^p_{\rm loc} (\R, \R)$ with $p \ge 2 $
such that $A = A_1 + A_2$, $V= V_1+V_2$, where $A_1,V_1$ have 
compact support and  $A_2, V_2 \in C^1(\R, \R)$ fulfill 
\begin{enumerate}[{\rm i$'$)}]
\item $A_2$ is supported away from $0$ and
          $ {\rm supp}(V_2) \subset  {\rm supp}(A_2)$,
\item $\|V_2/A_2\|_\infty <1$,
\item the derivative $(V_2/A_2)'$ is bounded on $\R$.
\end{enumerate}
\end{hypothesis}
Due to this assumptions we can choose $\beta = V_2/A_2=V_2/A$ and set
$\theta = \tanh^{-1} \beta$, $\gamma^{-2} =1-\beta^2$ as
above. Then, as in \eqref{them}, we obtain
\begin{align*}
\begin{split}
e^{-\sigma_2\theta/2 } h e^{\sigma_2\theta/2 }  
& = 
\gamma (1 -   \sigma_2 \beta) \sigma_1
\big(-\ri \partial_1 -\ri \sigma_2 \tfrac{\theta'}{2}\big) 
- \big(1-\sigma_2 \tfrac{V_2}{A}\big)\sigma_2A +V_1\\
&= M \Big[\sigma_1(-\ri \partial_1) -\sigma_2 A/\gamma 
+\gamma(1 + \sigma_2\beta)V_1
+ \sigma_3\tfrac{\theta'}{2} \Big]
\end{split}
\end{align*}
on $C_0^\infty(\R,\C^2)$, using again the notation 
$M =\gamma(1-\sigma_2\beta)$. Beside of the residual 
terms $\gamma(1 + \sigma_2\beta)V_1$,
$\sigma_3\theta'/2$, the operator 
\begin{align*}
 \hat h := \sigma_1(-\ri \partial_1) -\sigma_2 A/\gamma 
+\gamma(1 + \sigma_2\beta)V_1 
+ \sigma_3\tfrac{\theta'}{2} \quad 
\mathrm{on} \ L^2((0,\infty),\C^2)
\end{align*}
has a magnetic vector potential $A/\gamma$.
As in the case $A_2< V_2$, we conclude
\begin{lemma}
Consider the operator $h$ with potentials $A,V$ 
satisfying {\rm (H1$^\prime$)}. Then 
$\mathcal{D} (\hat h) = \mathcal{D} (M\hat h) 
=e^{-\sigma_2\theta/2 }\mathcal{D} (h)$ and the resolvent sets fulfill
$ \varrho( h)= \varrho(M\hat h)$. For any $z\in  \varrho(h)$ holds
\begin{align*}
&(M\hat h-z)^{-1}=e^{- \sigma_2\theta/2} (h-z)^{-1}e^{
  \sigma_2\theta/2}\quad \mbox{and}\\&\big\| (M\hat{ h}-z)^{-1}\big\| \le 
\big\|(h-\ri)^{-1}\big\|\big\| e^{|\theta|} \big\|_{\infty}.
\end{align*}
\end{lemma}
\begin{remark}
The same statement is valid for $h_k$ with corresponding
conditons and operators $\hat h_k$, $k \in \Z +\tfrac{1}{2}$.
\end{remark}
\section{Hilbert-Schmidt bounds}\label{proofhilbert}
In this section we prove Theorems \ref{mainlemma} and
\ref{hshk}. We have verified these results already  for the Dirac operator
with purely electric potentials in Proposition \ref{basicHS}. To treat the general case we use the connection
between $h$ ($h_k$) and $\tilde{h}$ ($\tilde{h}_k$) established in the previous section.

For the next Lemmas recall the definitions of the self-adjoint
operators $h$ and $h_k$ and of $\tilde{h}$ and $\tilde{h}_k$ given 
in Sections \ref{basic} and \ref{loro}, respectively.
\begin{lemma}\label{relativ}
Assume that $A,V$ satisfy Hypothesis {\rm (H1)}. Then,
there is a bounded operator $S\equiv S(A,V)$ such that 
\begin{align*}
  \big (\tilde{h}-\ri \big)^{-1} =U_2^*(-\ri\sigma_1\partial_x-\ri)^{-1} S, 
\end{align*}
where $[U_2 \psi] (x) = 
\exp \Big( {\ri \sigma_1 \int_0^x
 (V_2/\gamma)(s) \mathrm{d}s}\Big)\psi(x)$.
\end{lemma}
\begin{proof}
Recall that
\begin{align*}
  \tilde h=-\ri \sigma_1\partial_x 
-\gamma(1 + \sigma_2\beta)\sigma_2 A_1
+V_1/\gamma +V_2/\gamma+ \sigma_3 \tfrac{\theta'}{2},
\end{align*}
where $\theta =\tanh^{-1}(A_2/V)$ has a uniformly bounded derivative. For shorthand
notation we set $W:=-\gamma(1 + \sigma_2\beta)\sigma_2 A_1
+V_1/\gamma$. By the second
resolvent identity we obtain
\begin{align*}
  \big (\tilde{h}-\ri \big)^{-1}& =
  \left(-\ri\sigma_1\partial_x+W+V_2/\gamma-\ri\right)^{-1}\Big[1-
  \sigma_3
  \tfrac{\theta'}{2} \, \big(\tilde{h}-\ri\big)^{-1} \Big]\\
  & = U_2^*\left(-\ri\sigma_1\partial_x+\widetilde{W}-\ri\right)^{-1}U_2\Big[1-
  \sigma_3 \tfrac{\theta'}{2} \, \big(\tilde{h}-\ri\big)^{-1} \Big],
\end{align*}
where $\widetilde{W}=U_2 W U_2^*$. By the assumptions $|\widetilde{W}| \in L^p$ for some
$p\ge 2$, which implies that $\widetilde{W}$ is relatively compact with respect to
$-\ri\sigma_1\partial_x$ (use Kato-Seiler-Simon inequality).  In
particular, $\widetilde{W}(-\ri\sigma_1\partial_x+\widetilde{W}-\ri)^{-1}$ is
bounded. Finally, by the second resolvent identity we have
\begin{align*}
 \left(-\ri\sigma_1\partial_x+\widetilde{W}-\ri\right)^{-1}=
\big(-\ri\sigma_1\partial_x-\ri \big)^{-1}
\Big[1-   \widetilde{W}\left(-\ri\sigma_1\partial_x+\widetilde{W}-\ri\right)^{-1}\Big],
\end{align*}
from which follows the claim.
\end{proof}
\noindent
Throughout the rest of this section we write
$F:=\sigma_2\theta/2=\sigma_2 A_2/(2V)$ (see Section \ref{loro}).
\begin{proof}[Proof of Theorem \ref{mainlemma}]
  By a simple perturbational argument it suffices to prove the
  statement for $\xi=0$, i.e. for $h(0)=h$. Using Equation \eqref{them} and Lemma
  \ref{resolventestimate} we compute
\begin{align*}
  \mathbbm{1}_{I}(h-\ri)^{-1}&= \mathbbm{1}_{I}e^F (M\tilde
  h-\ri)^{-1} e^{-F}\\
  &= \mathbbm{1}_{I}e^F (\tilde h-\ri)^{-1}\Big[(\tilde h-\ri) (M\tilde h-\ri)^{-1}\Big]
  e^{-F},
\end{align*}
where the operator in $[...]$ is bounded by Lemma
\ref{resolventestimate} and the Closed Graph Theorem. Thus, we find
some constant $c$ such that
\begin{align*}
   \left\|\mathbbm{1}_{I}(h-\ri)^{-1}\right\|_{\rm{HS}} \le c
   \big\|\mathbbm{1}_{I}(\tilde h-\ri)^{-1}\big\|_{\rm{HS}}.
 \end{align*}
The claim is now a direct consequence of Lemmas \ref{relativ} and
Proposition \ref{basicHS}.
\end{proof}
\begin{proof}[Proof of Theorem \ref{hshk}]
Observe that 
\begin{align*}
   \mathbbm{1}_{I}(h_k-\ri)^{-1}
   =e^F\mathbbm{1}_I(\tilde{h}_k-\ri)^{-1} \Big[(\tilde{h}_k-\ri)
(M\tilde{h}_k-\ri)^{-1}e^{-F} \Big],
\end{align*}
which implies, by the Closed Graph Theorem and  Lemma \ref{mainlemma}, that
\begin{align}\label{i1}
  \left\| \mathbbm{1}_{I}(h_k-\ri)^{-1} \right\|_{\rm{HS}} \le 
 c \big\| \mathbbm{1}_{I}(\tilde{h}_k-\ri)^{-1} \big\|_{\rm{HS}}, 
\end{align}
for some constant $c>0$. Here 
\begin{equation*}
\tilde h_k = 
\sigma_1(-\ri \partial_x) 
+\gamma(1 + \sigma_2\beta)\sigma_2 \big(\tfrac{k}{x} -A_1\big)
+V_1/\gamma +V_2/\gamma+ \sigma_3 \tfrac{\theta'}{2}.
\end{equation*}
In order to make the argument more transparent we write
\begin{align*}
  \tilde h_k = 
\sigma_1(-\ri \partial_x) 
+\sigma_2\tfrac{k}{x} +W
+V_2/\gamma,
\end{align*}
where $W=W_1+W_2$ and 
\begin{align*}
  W_1:=V_1/\gamma -\gamma(1 + \sigma_2\beta)\sigma_2 A_1,
\end{align*}
which has compact support and  is obviously in $L^p(\R^+,\C^{2\times
  2})$ for some $p>2$, and
\begin{align*}
  W_2:=\gamma(1 + \sigma_2\beta)\sigma_2
  \tfrac{k}{x}-\sigma_2\tfrac{k}{x} +\sigma_3
  \tfrac{\theta'}{2}=\Big((\gamma-1)+\gamma \beta \sigma_2\Big)
\sigma_2\tfrac{k}{x} +\sigma_3
  \tfrac{\theta'}{2}.
\end{align*}
Due to the support properties of $A_2$ (recall that
$\beta=A_2/V$ and $\gamma=(1-\beta^2)^{-1/2}$), the function $W_2$ is
supported away from zero. In addition, observe that $W_2$ is uniformly bounded
and hence
\begin{align}\label{i2a}
  \big\| \mathbbm{1}_{I}(\tilde{h}_k-\ri)^{-1} \big\|_{\rm{HS}}\le  
c \big\| \mathbbm{1}_{I}(\tilde{h}_k-W_2-\ri)^{-1} \big\|_{\rm{HS}}.
\end{align}
According to Corollary \ref{corol2} (from Section \ref{last}) $W_1$ is an
infinitesimally small perturbation with respect to $\tilde h_k -W=
\sigma_1(-\ri \partial_x) +\sigma_2\tfrac{k}{x} +V_2/\gamma$ and therefore
\begin{align}\label{i2}
  \big\| \mathbbm{1}_{I}(\tilde{h}_k-W_2-\ri)^{-1} \big\|_{\rm{HS}}\le  
c \big\| \mathbbm{1}_{I}(\tilde{h}_k-W-\ri)^{-1} \big\|_{\rm{HS}}.
\end{align}
Let us define the self-adjoint operator
$$h_0:=\sigma_1(-\ri \partial_x) +V_2/\gamma$$ with domain,
$\mathcal{D}(h_0)$, given by \eqref{domain1} for $\alpha=0$. We can
compare the resolvents of $\tilde{h}_k-W$ and $h_0$ as follows: 

Define $\chi\in C^\infty((0,\infty),[0,1])$ such that $\chi=0$ on
$(0,\tfrac{1}{2})$ and $\chi=1$ on $[1,\infty)$. We compute
\begin{align*}
  \mathbbm{1}_{I}\big(\tilde{h}_k-W-\ri\big)^{-1} &= \mathbbm{1}_{I}\chi
  \big(\tilde{h}_k-W-\ri\big)^{-1}\\
  &= \mathbbm{1}_{I} (h_0-\ri)^{-1}\Big[ (h_0-\ri)\chi
  \big(\tilde{h}_k-W-\ri\big)^{-1}\Big].
\end{align*}
Observe that the operator in $[...]$ is bounded by Remark \ref{ws2}
and the Closed Graph Theorem  (and its norm will
depend on $|k|$). Thus there is a $c_{|k|}$ such that 
\begin{align}\label{i3}
\left\| \mathbbm{1}_{I}(\tilde{h}_k-W-\ri)^{-1} \right\|_{\rm{HS}}\le 
c_{|k|} \left\|\mathbbm{1}_{I} (h_0-\ri)^{-1}\right\|_{\rm{HS}}\,.
\end{align}
This implies the result by Proposition
\ref{basicHS}.
\end{proof}
\section{Proof of Theorem \ref{lastmainthm2}}\label{last}
The main object of this section is the proof of Theorem
\ref{lastmainthm2} for the Dirac operators $h_k$. As we already
mentioned $h_k$
needs special care due to the $k/x$-singularity. Let us explain this a little further: Recall that for $A=V=0$
\begin{align}\label{amelie2}
h_k = \sigma_1 (-\ri \partial_x) + \sigma_2\tfrac{k}{x} 
\quad \mbox{on} \quad L^2((0, \infty), \C^2). 
\end{align}
Observe that, for $\varphi \in C^\infty_0((0,\infty),\C^2)$,
\begin{equation}
\label{hhhh}
\begin{split}
\|\varphi\|_{h_k}^2 &= 
\|\varphi'\|^2 +\|\varphi\|^2 - \SP{\varphi}{\sigma_3\frac{k}{x^2}\varphi}
+ \left\| \frac{k}{x} \varphi\right\|^2 \\ &\ge
\|\varphi'\|^2 + \|\varphi\|^2  +
(k^2 -|k|)\left\| \frac{1}{x} \varphi\right\|^2.
\end{split}
\end{equation}
Using this and the standard Hardy inequality on the half-line
\begin{align}
  \label{eq:6}
  \int_0^\infty |\varphi'(x)|^2 dx \ge \frac{1}{4}  \int_0^\infty \frac{|\varphi'(x)|^2}{x^2} dx
\end{align}
one can show, for the case $|k|>1/2$, that $L^p$-perturbations can be
controlled by the $h_k$-graph norm, whenever $p\ge 2$. For the
important case $|k|=1/2$, this argument does not work since the Hardy
inequality becomes critical. Instead, a version of the
Hardy-Sobolev-Maz'ya inequality on the half-line, proven recently in
\cite{FrankLoss2012}, allows us to control $L^p$-perturbations, but
only for $p> 2$. These observations, made precise in Theorem
\ref{perturbhk} and Corollary \ref{corol2}, enable us to show that the
operator $\mathbbm{1}_{(0,R)}(h_k-\ri)^{-1}$ is compact, provided the
potentials $A,V$ are locally in $L^p$ for $p> 2$ (Corollary
\ref{corol1}). The latter is an important ingredient for the proof of
Theorem \ref{lastmainthm2} at the end of this section.
\begin{remark}\label{hardy}
In view of \eqref{hhhh} and  \eqref{eq:6}
it is clear that $1/x$ is a perturbation with respect to $\big(-\ri
  \sigma_1\partial_x + \sigma_2\tfrac{k}{x}\big)$, provided the strict
  inequality $|k|>1/2$ holds.
\end{remark}
\begin{theorem}\label{perturbhk}
For  $|k| \ge \tfrac{1}{2}$ consider $h_k$ with $A=V=0$. Then, any
multiplication operator $M \in L^p((0, \infty), \C^{2\times 2})$ with
$p>2$ is infinitesimally $h_k$-bounded. In addition, 
any multiplication operator $M \in L^2((0, \infty), \C^{2\times 2})$ 
is infinitesimally $h_k$-bounded for $|k| >\tfrac{1}{2}$.
\end{theorem}
\begin{proof}
Let $\varphi \in C^\infty_0((0,\infty),\C^2)$ and $|k|\ge \tfrac{1}{2}$, then 
\begin{align}\nonumber
\|\varphi\|_{h_k}^2 &= 
\|\varphi'\|^2 +\|\varphi\|^2 - \SP{\varphi}{\sigma_3\frac{k}{x^2}\varphi}
+ \left\| \frac{k}{x} \varphi\right\|^2 \\\label{jakelin} &\ge
\|\varphi'\|^2 + \|\varphi\|^2  +
(k^2 -|k|)\left\| \frac{1}{x} \varphi\right\|^2 \\\nonumber & \ge
\|\varphi'\|^2 +\|\varphi\|^2 - e^{-4(|k|-\frac{1}{2})^{2}} 
\frac{1}{4}\left\| \frac{1}{x} \varphi\right\|^2. 
\end{align}
Using the one-dimensional Sobolev inequality $\|\varphi\|_\infty^2\le
\kappa \|\varphi\|^2+\kappa^{-1}\|\varphi'\|^2$ valid for $\kappa>1$,
we get that
\begin{align}\label{joseflikeslatinas}
  \|\varphi\|_{h_k}^2&\ge
\big(1- \mu(k) \big)\kappa \|\varphi\|_\infty^2+
(1-\kappa^{2})\|\varphi\|^2  + 
\mu(k)\left(
\|\varphi'\|^2 -\frac{1}{4}\left\| \frac{1}{x} \varphi\right\|^2 \right),
\end{align}
where $\mu(k) :=e^{-4(|k|-\frac{1}{2})^{2}} \in (0,1]$. (Note that the
first term on the right hand side of \eqref{joseflikeslatinas} equals
zero when $k=1/2$.) By the Hardy-Sobolev-Maz'ya inequality on the
half-line (see \cite[Thm. 1.2]{FrankLoss2012}) we obtain, for $q \in
(2, \infty)$ and $\theta =\frac{1}{2}(1-2q^{-1})$, a constant
$c_\theta$  (depending only on $\theta$) such that
\begin{align*}
c_\theta\|\varphi\|^2_q  &\le \left(\|\varphi'\|^2 -
\frac{1}{4}\left\| \frac{1}{x} \varphi\right\|^2\right)^\theta
\big(\|\varphi\|^2\big)^{1-\theta}\\ 
&\le\epsilon \theta\left(\|\varphi'\|^2 -
\frac{1}{4}\left\| \frac{1}{x} \varphi\right\|^2\right) +
(1-\theta)\epsilon^{-\frac{\theta}{1-\theta}}\|\varphi\|^2. 
\end{align*}
In the last step we use Young's inequality with $\epsilon\in (0,1)$.
Combining this  with \eqref{joseflikeslatinas} we
conclude that
\begin{align}\label{edgardolikesperuvianfood}
\begin{split}
\mu(k)c_\theta\|\varphi\|^2_q +
(1-\mu(k)) \kappa&\epsilon \theta \|\varphi\|_\infty^2 \le \epsilon \theta\|\varphi\|_{h_k}^2 + 
c(\epsilon,\kappa,\theta)\|\varphi\|^2.
\end{split}
\end{align}
For  $M \in L^p((0, \infty), \C^{2\times 2})$ with $p>2$ 
we choose $\theta = p^{-1}$ (hence $p^{-1}+q^{-1}=1/2$), then 
\eqref{edgardolikesperuvianfood} yields
\begin{align}\label{carrete1}
\|M\varphi\|^2 \le \|M\|_p^2\|\varphi\|_q^2 \le
\|M\|^2_p (\mu(k) c_\theta)^{-1}\left( 
\epsilon\theta\|\varphi\|_{h_k}^2 +
c(\epsilon,\kappa,\theta)\|\varphi\|^2 \right)
\end{align}
for any $\epsilon \in (0, 1]$. If $M$ is a $L^2$-function and $|k| > \tfrac{1}{2}$
we use again \eqref{edgardolikesperuvianfood} (dropping the
first term) to obtain
\begin{align}\label{carrete2}
\|M\varphi\|^2 \le \|M\|_2^2\|\varphi\|_\infty^2 \le
\frac{\kappa^{-1}}{(1-\mu(k))}\|M\|_2^2
\|\varphi\|_{h_k}^2 
+\tilde{c}(\epsilon,\kappa,\theta)\|\varphi\|^2.
\end{align}
 Since $h_k$ is essentially
 self-adjoint on $C^\infty_0((0,\infty),\C^2)$ inequalities
\eqref{carrete1} and \eqref{carrete2} imply the claim.
\end{proof} 
\begin{corollary}\label{corol2}
For  $|k| \ge \tfrac{1}{2}$ consider $h_k$ with $A,V \in L_{\rm loc}^p((0, \infty), \R)$ for some $p>2$.
Then any multiplication operator 
$M \in L^{s}((0, \infty), \C^{2\times 2})$, $s>2$, with compact support
is infinitesimally $h_k$-bounded.
\end{corollary}
\begin{proof}
Let $\chi\in C^\infty(\R^+,[0,1])$ be a smooth cut-off function
which equals $1$ on the support of $M$ and vanishes for large $x$. Then, for any $\varphi
\in C_0^\infty(\R^+,\C^2)$ and $\epsilon\in (0,1)$ we find, by Theorem \ref{perturbhk}, a constant
$c_\epsilon$ such that 
\begin{align}\label{rose}
  \|M\varphi\|=\|M\chi\varphi\|\le \epsilon \big\|\big (-\ri
  \sigma_1\partial_x + \sigma_2\tfrac{k}{x}\big)\chi \varphi \big\|+ c_\epsilon \|\varphi\|.
\end{align}
Let us write $W:=V-\sigma_2 A \in L_{\rm loc}^p((0, \infty),
\R)$. Using again Theorem \ref{perturbhk} 
we find a constant $c>0$ with 
 \begin{align*}
\big\|\big (-\ri \sigma_1\partial_x + \sigma_2\tfrac{k}{x}\big)\chi
\varphi \big\|&\le \|h_k\chi\varphi\|+\|W\chi\varphi\|\\
&\le \|h_k\chi\varphi\|+\tfrac{1}{2}
\big\|\big (-\ri \sigma_1\partial_x + \sigma_2\tfrac{k}{x}\big)\chi
\varphi \big\|+c\|\varphi\|\\
&\le \|h_k\varphi\|+\tfrac{1}{2}
\big\|\big (-\ri \sigma_1\partial_x + \sigma_2\tfrac{k}{x}\big)\chi
\varphi \big\|+(c+\|\chi'\|_\infty)\|\varphi\|.
\end{align*}
We get the desired  result by combining this with \eqref{rose}. 
\end{proof}
\begin{corollary}\label{corol1}
For  $|k| \ge \tfrac{1}{2}$ consider $h_k$ with $A,V \in L_{\rm loc}^p((0, \infty), \R)$ for some $p>2$.
Then, $h_k$ is a locally compact operator,
i.e. for any $R >0$ the operator $\mathbbm{1}_{(0,R)}(h_k-\ri)^{-1}$
is compact.
\end{corollary}
\begin{proof}
  For $R >0$ let $\chi \in
  C^\infty([0,\infty),[0,1])$ be a smooth cutoff-function with
  $\chi  (x) =1$ for $x\le R$ and $\chi  (x)=0$ for
  $x\ge R+1$. We compare $h_k$ with the reference operator
\begin{align*}
h_{\rm ref} = \sigma_1(-\ri \partial_x) +\sigma_2\frac{1}{x}
\quad \mbox{on} \quad L^2((0, \infty), \C^2),
\end{align*}
which is known to be locally compact
(see e.g. \cite{KMS1995}).
We compute the resolvent difference 
\begin{align*}
\chi ^2 \frac{1}{h_{\rm ref}-\ri} - \frac{1}{h_k-\ri}\chi ^2 
=&\ \frac{1}{h_k-\ri}\left( (h_k-\ri) \chi ^2 -
\chi ^2 (h_{\rm ref}-\ri)\right) \frac{1}{h_{\rm ref}-\ri} \\
=&\ \frac{1}{h_k-\ri}\left( (V-\sigma_2 A) \chi ^2  
-2\ri \sigma_1 \chi \chi ' \right)  \frac{1}{h_{\rm ref}-\ri}  \\
& + (k-1)\frac{1}{h_k-\ri}x^{-1/4}
\chi ^2\sigma_2 x^{-3/4}\frac{1}{h_{\rm ref}-\ri}. 
\end{align*}
Using that $h_{\rm ref}$ is locally compact and that $(V-\sigma_2 A) \chi $,
$x^{-1/4}\chi $, and $2\ri \sigma_1 \chi'$ are relatively
$h_k$-bounded (see Corollary \ref{corol2}), it suffices to show that
\begin{align*}
\chi  x^{-3/4}\frac{1}{h_{\rm ref}-\ri} 
\end{align*}
is a compact operator. To this end we first recall that $x^{-2}$ is bounded
with respect to $h_{\rm ref}^2$ in the sense of quadratic forms (see
Remark \ref{hardy}). Since exponentiating to the power $3/4$ is
operator monotonic we conclude that $x^{-3/4}|h_{\rm
  ref}-\ri|^{-{3/4}}$ is bounded.  Therefore, by the relation
\begin{align*}
\chi  x^{-3/4}\frac{1}{h_{\rm ref}-\ri} = \ &
x^{-3/4}\left|\frac{1}{h_{\rm ref}-\ri}\right|^{3/4}
{\rm sgn} \left(\frac{1}{h_{\rm ref}-\ri} \right)
\left|\frac{1}{h_{\rm ref}-\ri}\right|^{1/4}\chi  \ +\\
&x^{-3/4}\frac{1}{h_{\rm ref}-\ri}\big(-\ri \sigma_1\chi '\big) 
\frac{1}{h_{\rm ref}-\ri}
\end{align*}
it suffices to show that $\chi |h_{\rm ref}-\ri|^{-1/4}$ is
compact. This, however, follows from the identity
\begin{align*}
\left|\frac{1}{h_{\rm ref}-\ri}\right|^{1/4}\chi  &=
\left(\frac{1}{h_{\rm ref}^2+1}\right)^{1/8}\chi  \\ &=
B(\tfrac{7}{8},\tfrac{1}{8})^{-1}\int_0^\infty \frac{1}{h_{\rm
    ref}^2+s+1} \chi \,
\frac{1}{s^{1/8}} \,{\rm d}s\,  
\end{align*}
(here $B(x,y)$ denotes the beta function), since a Riema\-nnian integral of
compact operators is also compact.
\end{proof}
\begin{remark}\label{rm1}
If we assume further that $|k|>1/2$ then, 
in view of Theorem \ref{perturbhk}, the
statements of Corollaries \ref{corol2} and  \ref{corol1} are also
valid  for $L^2$-perturbations.
\end{remark} 
\begin{proof}[Proof of Theorem \ref{lastmainthm2}]
In this proof we slightly modify the argument of  \cite[Theorem
6.2]{Last1996} for the operator defined on the half-line.
Let $\mathcal{B}(\R)$ denote the Borel $\sigma$-algebra on
$\R$. Given any $\psi \in L^2(\R^+, \C^2)$ we write its
associated  spectral measure (with respect to $h_k$) as  
\begin{align*}
 \mu_{\psi} : 
\mathcal{B}(\R) \to [0, \infty), \qquad
\Omega \mapsto \Sps{\psi}{\mathbbm{1}_\Omega(h_k)\psi}. 
\end{align*}
Since $\mu_\psi$ is absolutely continuous with respect to the Lebesgue
measure it can be decomposed as a sum of mutually singular measures
$\mu_\psi=\mu_{\psi,1}+\mu_{\psi,2}$, where
$\mu_{\psi,2}(\R)<\|\psi\|^2/4$ and $\mu_{\psi,1}$ is a uniformly
Lipschitz continuous measure, i.e. there is a constant $C>0$ such that
for any interval with Lebesgue measure $|I|<1$, $\mu_{\psi,1}(I)<C|I|$
(this can be verified decomposing the Radom-Nykodym derivative,
$f_\psi$, associated to $\mu_{\psi}$ as
$f_\psi=f_\psi\mathbbm{1}_{\{f_\psi<\alpha\}}
+f_\psi\mathbbm{1}_{\{f_\psi>\alpha\}} $ for $\alpha>0$ sufficiently
large; for a more general statement involving uniform
$\alpha$-H\"older continuity see \cite[Theorem 4.2]{Last1996}).  
For $j\in\{1,2\}$ define $\psi_j:=\mathbbm{1}_{S_j}(h_k)\psi$ where
$S_j\subset \R$ is the support of the measure $\mu_{\psi,j}$. Then,
for any $\Omega \in \mathcal{B}(\R) $, we have
\begin{align*}
  \mu_{\psi_j}(\Omega)&=\Sps{\psi_j}{\mathbbm{1}_\Omega(h_k) \psi_j}
=\Sps{\psi}{\mathbbm{1}_{\Omega\cap S_j}(h_k) \psi}\\
&=\mu_{\psi}(\Omega\cap S_j)=\mu_{\psi,1}(\Omega\cap S_j)
+\mu_{\psi,2}(\Omega\cap S_j)= \mu_{\psi, j}(\Omega),
\end{align*}
where in the last equality we use that the measures $\mu_{\psi,j}$ are
disjointly supported. Thus, we get that $\mu_\psi=\mu_{\psi_1}+
\mu_{\psi_2}$.  Note that $\psi_1\not=0$ since
$$\|\psi_1\|^2=\mu_{\psi_1}(\R)=\|\psi\|^2-
\mu_{\psi_2}(\R)\ge 3\|\psi\|^2 /4.$$
For any $R>1$ we have
\begin{align}\label{eq:5}
  \| x^{p/2} e^{-\ri h_k t}\psi \|^2 \ge
\| R^{p/2} \mathbbm{1}_{ (R, \infty )}e^{-\ri  h_k t}\psi \|^2 \ge
R^{p} \big( \|\psi\|^2 - \| \mathbbm{1}_{ (0,R )} e^{-\ri  h_k t}\psi \|^2 \big).
\end{align}
We observe that since  $\psi_1$ is orthogonal to $\psi_2$  the triangular
inequality  yields
\begin{equation}
  \label{eq:4}
\begin{split}
\|\mathbbm{1}_{ (0,R )} e^{-\ri  h_k t}\psi \|^2&\le
2\|\mathbbm{1}_{ (0,R )} e^{-\ri  h_k t}\psi_1\|^2+2\|\psi_2\|^2\\&\le 
2\|\mathbbm{1}_{ (0,R )} e^{-\ri  h_k t}\psi_1\|^2+\tfrac{1}{2}\|\psi\|^2.
\end{split}
\end{equation}
In order to use Theorem \ref{hshk} we replace the cut-off function above
by one supported away from zero. Note that by Corollary \ref{corol1}
and the RAGE theorem we find
a $T_0>0$ such that for all $T>T_0$ one has that  
\begin{align*}
  \langle\| \mathbbm{1}_{ (0,R )} e^{-\ri  h_k t}\psi_1
  \|^2\rangle_T&=\langle\| \mathbbm{1}_{ (0,1)} e^{-\ri  h_k t}\psi_1
  \|^2\rangle_T+\langle\| \mathbbm{1}_{ (1,R)} e^{-\ri  h_k t}\psi_1
  \|^2\rangle_T\\
&\le \tfrac{1}{8}\|\psi\|^2+\langle\| \mathbbm{1}_{ (1,R)} e^{-\ri  h_k t}\psi_1
  \|^2\rangle_T.
\end{align*}
Combining the latter bound with \eqref{eq:4} and \eqref{eq:5} we
readily obtain, for $T>T_0,$
\begin{align}
  \label{eq:7}
  \begin{split}
\big\langle  \| x^{p/2} e^{-\ri h_k t}\psi \|^2 \big\rangle_T \ge R^{p} 
\big( \tfrac{1}{4}\|\psi\|^2 - 2 \big\langle\| \mathbbm{1}_{ (1,R )} e^{-\ri  h_k t}\psi_1 \|^2 \big\rangle_T \big).
  \end{split}
\end{align}
Next we recall (see \cite[Theorem 3.2]{Last1996})  that given a
self-adjoint 
operator $H$ and  a Hilbert-Schmidt operator $A$ one finds a
constant $c_\varphi$ such that 
\begin{align}\label{Last3.2}
\langle \|Ae^{-\ri t H}\varphi\|\rangle_T\le c_\varphi \|A\|^2_{\rm
  HS}T^{-1}
\end{align}
provided the $H$-spectral measure associated to $\varphi$ is Lipschitz continuous.
Applying this and  Theorem \ref{hshk} we obtain
\begin{align}
\begin{split}
  \label{eq:31}
  \langle\| \mathbbm{1}_{ (1,R)} e^{-\ri  h_k t}\psi_1
  \|^2\rangle_T&= \langle\| \mathbbm{1}_{ (1,R)} \mathbbm{1}_\triangle(h_k) e^{-\ri  h_k t}\psi_1
  \|^2\rangle_T\\
&\le c_{\psi_1}T^{-1}\| \mathbbm{1}_{ (1,R)}(h_k-\ri)^{-1}
(h_k-\ri)\mathbbm{1}_\triangle(h_k) \|_{\rm HS}^2 \\
&\le c_{\psi_1} c_{\triangle}  C_k   R T^{-1}
\equiv \tfrac{1}{2} \widehat C_k(\psi_1, \Delta, k) R  T^{-1},
\end{split}
\end{align}
where $c_{\triangle} =\|(h_k-\ri)\mathbbm{1}_\triangle\|^2$. Hence, using the latter bound  in
\eqref{eq:7} we get, for $T>T_0$, that
\begin{align*}
  \big\langle  \| x^{p/2} e^{-\ri h_k t}\psi \|^2 \big\rangle_T &
\ge R^{p}  \big( \tfrac{1}{4}\|\psi\|^2
 - \widehat C_k(\psi_1, \Delta, k) R T^{-1} \big)\\
&=\frac{1}{8^{p+1}\widehat C_k{(\psi_1, \Delta, k)}^p}\|\psi\|^{2p+2}\, T^{p},
\end{align*}
where in the last equality we have chosen
\begin{align*}
  R\equiv R(T)=\frac{\|\psi\|^2}{8\widehat C_k(\psi_1, \Delta, k) } \ T.
\end{align*}
Finally note that the inequality \eqref{theineq} is trivially
fulfilled for finite $T\in [0,T_0]$ by just choosing the constant
$ C_k(\psi_1, \Delta, k) $ suitably. \vspace{1cm}
\end{proof}
\noindent
{\bf Acknowledgments.}
The authors want to thank Jean-Marie Barbaroux, Jean-Claude Cuenin and
Karl-Michael Schmidt for useful discussions and remarks.  J.M. also
likes to thank the {\it Faculdad de F\'isica de la Pontificia
  Universidad Cat\'olica de Chile} for the hospitality during his
research stay.  J.M. has been supported by SFB-TR12 ``Symmetries and
Universality in Mesoscopic Systems" of the DFG.  E.S. has been
supported by Fondecyt (Chile) project 1141008 and Iniciativa
Cient\'ifica Milenio (Chile) through the Millenium Nucleus RC–120002
``F\'isica Matem\'atica” .
\begin{appendix}
\section{Remarks on Self-adjointness of Dirac operators}\label{s.a.}
In this section we state and prove some facts concerning
self-adjointness of the one-dimensional Dirac operators discussed in
Section \ref{basic}. These facts are well known, however, they are not easy
to find in the standard literature. 
\begin{proposition}\label{lpcin0}
Let $|k| \ge \tfrac{1}{2}$, then $h_k$ is in the limit point case at $0$.
\end{proposition}
\begin{proof}
It suffices to show that there is a solution to the eigenvalue problem
\begin{align}\label{evp}
  h_k\varphi=\lambda\varphi
\end{align}
which is not square-integrable at $0$. According to 
\cite[Theorem1]{EasthamSchmidt2008}
 (see also \cite{Titchmarsh1961}), for $k\ge\tfrac{1}{2}$, 
there is a unique solution to \eqref{evp} with the 
asymptotic behaviour
\begin{align*}
  u(r)=(o(1), 1+o(1))^{\rm T}x^{k}\quad\mbox{as}\quad x\to 0.
\end{align*}
(Note that \cite[Theorem 1]{EasthamSchmidt2008} is only stated for the
case $A=0$, however, the same argument applies provided $A$ is
integrable at zero.) Let $w$ be a linear independent solution of
\eqref{evp} such that the Wronski determinant
$W(u,w):=u_1w_2-u_2w_1\equiv 1$. Assume that $\liminf_{x\to
  0} |w(x)|x^{k}=0$, then clearly $\liminf_{x\to 0} W(u,w)(x)=0$ which is
  a contradiction.  Hence $w$ can not be square-integrable at $0$.  An
  analogous argument holds also when $k\le -1/2$.
\end{proof}
\begin{lemma}\label{good-core}
  Let $k\in \Z+\tfrac{1}{2}$, then $C_0^\infty((0,\infty),\C^2)$ is dense in
  $\mathcal{D}_0(h_k)$ with respect to the $h_k$-graph norm. The
  analogous statement holds for $C_0^\infty(\R,\C^2)\subset
  \mathcal{D}_0(h)$ in the $h$-graph norm.
\end{lemma}
\begin{proof}
  We give the details of the proof only for the operator defined on the whole real
  line. First note that clearly $C_0^\infty(\R,\C^2)$ is a subset of
  $\mathcal{D}_0(h)$. Let $\psi \in \mathcal{D}_0(h)$ and let $K$
  be a compact set which contains the support of $\psi$. Since $\psi\in
  C_0(\R,\C^2)$ and $(V-\sigma_2 A)\in L^2_{\rm
  loc}$ we have that 
$$ -\ri \sigma_1\psi'=h\psi-(V-\sigma_2 A)\psi \in L^2(\R,\C^2),$$
which implies that $\psi\in H^1(\R,\C^2)$.
Let $(\psi_n)_{n\in\N}\subset C_0^\infty$ be a sequence of  mollifiers
of $\psi$ whose support is also contained in $K$. We estimate
\begin{align*}
  \|\psi-\psi\|_h^2&=\|\psi-\psi_n\|_2^2+\|h(\psi-\psi_n)\|_2^2\le
  \|\psi-\psi_n\|_{H^1}^2
  +\|(V-\sigma_2 A)(\psi-\psi_n)\|_2\\
  &\le\|\psi-\psi_n\|_{H^1}^2+\|\psi-\psi_n\|_\infty^2\,\|(V-\sigma_2
  A)\mathbbm{1}_{K}\|_2^2.
\end{align*}
By the $H^1$-convergence of mollifiers we know that
$\|\psi-\psi_n\|_{H^1}^2\to 0$. Moreover, the Sobolev inequality in
dimension one implies that $\|\psi-\psi_n\|_\infty^2\to 0$. Hence,
$\psi_n$ converges to $\psi$ in the graph norm, as claimed. The
argument for the operator $h_k$ is completely analogous. Just note
that $k/x \in L^2_{\rm loc}((0,\infty))$.
\end{proof}
\section{Computation of a resolvent kernel}\label{expres}
In order to compute a resolvent kernel of the operator 
$\sigma_1(-\ri \partial_x)$ on $L^2((0, \infty), \C^2)$ with the 
boundary condition $\psi_1(0) =0$, we use the unitary matrix 
$\hat U = \tfrac{1}{\sqrt{2}}(\mathbbm{1} +\ri \sigma_3)$ to 
transform the problem to the operator
\begin{equation*}
\sigma_2(\ri \partial_x) = 
\frac{1}{2} (\mathbbm{1} + \ri \sigma_3)
\sigma_1(-\ri \partial_x)
(\mathbbm{1} - \ri \sigma_3) =
\hat U \sigma_1(-\ri \partial_x) \hat U^*
\end{equation*}
on $L^2((0, \infty), \C^2)$ with the same boundary conditions.
The corresponding relation for the resolvent is
\begin{equation}\label{resolventapp}
\frac{1}{\sigma_1(-\ri\partial_x) -\ri}   =  
\hat U^*  \frac{1}{\sigma_2(\ri\partial_x) -\ri} \, \hat U.
\end{equation}
We note that, by \cite[Section 15.5]{Weidmann2}, the 
kernel of the resolvent \eqref{resolventapp} 
can be  given in terms of a fundamental
system  $u_1(z,\,\cdot \,), u_2(z,\,\cdot \,)$ of the ODE
$(\sigma_1(-\ri\partial_x) -z)u =0$. For $z \in \C \setminus \R$ 
this is given as follows 
\begin{equation*}
\frac{1}{\sigma_2(\ri\partial_x) -z} (x_1,x_2) =
\begin{cases}
\sum_{j,k =1}^2 m_{j,k}^+(z) \overline{u_j(\overline z,
  x_1)}u_k ^{\rm T} (z,x_2) 
& \mathrm{if} \; x_1 > x_2>0, \\[0.1cm]
\sum_{j,k =1}^2 m_{j,k}^-(z) \overline{u_j(\overline z, x_1)}u_k ^{\rm T}(z,x_2)
& \mathrm{if} \;  x_2 > x_1 >0,
\end{cases}
\end{equation*}
where
and $m^+(z), m^-(z)$ are $2\times 2$ matrices whose coefficients are
given in terms of certain complex numbers $m_a(z), m_b(z)$ that  are
chosen such that
$m_a(z) u_1(z,\,\cdot \,)+ u_2(z,\,\cdot \,)$ fulfill the boundary
condition at $0$ and $m_b(z) u_1(z,\,\cdot \,)+ u_2(z,\,\cdot \,)$
is square-integrable at $\infty$ (see \cite[Section
15.5]{Weidmann2}). Hence, using the fundamental system
\begin{equation*}
u_1(z,x) = 
\begin{pmatrix}
\cos zx \\ \sin zx
\end{pmatrix}, \quad \quad \quad
u_2(z,x) = 
\begin{pmatrix}
-\sin zx  \\ \cos zx
\end{pmatrix},
\end{equation*}
we obtain that $m_a(\ri) =0, m_b(\ri) = \ri$. Therefore, 
\begin{equation*}
m^+(\ri) = 
\begin{pmatrix}
0 & -1 \\
0 & \ri
\end{pmatrix}, \quad \quad \quad
m^-(\ri) = 
\begin{pmatrix}
0 & 0 \\
-1 & \ri
\end{pmatrix}.
\end{equation*}
 In addition, using that
\begin{equation*}
u_1(i,x) = \overline{u_1(-\ri,x)} =
\begin{pmatrix}
\cosh x \\ \ri \sinh x
\end{pmatrix}, \quad \quad \quad
u_2(i,x) = \overline{u_2(-\ri,x)} =
\begin{pmatrix}
-\ri \sinh x  \\ \cosh x
\end{pmatrix},
\end{equation*}
we get the explicit resolvent kernel
\begin{equation*}
\frac{1}{\sigma_2(\ri\partial_x) -\ri} (x_1,x_2) =
\begin{cases}
e^{-x_1} 
\begin{pmatrix}
\ri \sinh x_2 & -\cosh x_2 \\
\sinh x_2 & \ri \cosh x_2 
\end{pmatrix}
& \mathrm{if} \ x_1 > x_2>0, \\[0.5cm]
e^{-x_2} 
\begin{pmatrix}
\ri \sinh x_1 & \sinh x_1 \\
-\cosh x_1 & \ri \cosh x_1 
\end{pmatrix}
& \mathrm{if} \ x_2 > x_1>0.
\end{cases}
\end{equation*} 
Finally, by \eqref{resolventapp} we get \eqref{resolventkernel}.
\section{Proof of Corollaries \ref{appl1} and \ref{appl2}}\label{proofappl}
We observe that in the case of a translation symmetry 
in $x_2$-direction, we can write
\begin{align*}
\langle\| |x_1|^{p/2}e^{-\ri tH } \psi\|^2\rangle_T =
\int_{-\infty}^\infty \big \langle \big\| |x_1|^{p/2}
e^{-\ri th(\xi) } \widehat\psi (\, \cdot \, , \xi)
\big \|^2\big\rangle_T \, {\rm d}\xi
\end{align*}
By assumption \eqref{condappl1} we find $\xi_0>0$ large
enough such that  
\begin{align*}
M:= \big\{\xi \in [-\xi_0,\xi_0] \,|\, 
\widehat\psi (\, \cdot \, , \xi) \neq 0, \
\widehat\psi (\, \cdot \, , \xi) \in P_{ac}(h(\xi)) L^2(\R, \C^2)\big\}
\end{align*}
has non-zero Lesbegue measure. Using that 
\begin{align*}
\mathbbm{1}_\Delta(H) = 
\int_\R^\oplus \mathbbm{1}_\Delta (h(\xi)) {\rm d} \xi
\end{align*}
(see \cite[Theorem XIII.85]{Reed_simon_4}),  we conclude  
$\widehat\psi (\, \cdot \, , \xi) \in 
\mathbbm{1}_\Delta (h(\xi))  L^2(\R, \C^2)$ for a.e. $\xi\in \R$
if $\psi \in \mathbbm{1}_\Delta (H) L^2(\R^2, \C^2)$. Hence,
Theorem \ref{lastmainthm1} implies
\begin{align*}
\langle\| |x_1|^{p/2}e^{-\ri t H } \psi\|^2\rangle_T \ge
\int_M\big \langle \big\| |x_1|^{p/2}
e^{-\ri th(\xi) } \widehat\psi (\, \cdot \, , \xi)
\big \|^2\big\rangle_T \, {\rm d}\xi \ge
T^p \int_M C_\xi (\psi, \Delta ,p) {\rm d}\xi \,,
\end{align*}
which give us the desired bound since
$\int_M C_\xi (\psi, \Delta ,p) {\rm d}\xi =: C (\psi, \Delta ,p) > 0$.

Now let us consider the case, when $H$ is spherically symmetric
and fulfills the assumptions of  Corollary \ref{appl1}. We write
\begin{align*}
 H \cong \bigoplus_{k \in \, \Z+\frac{1}{2}}  h_k,
\end{align*}
and use for 
\begin{align}\label{kunigunde}
\psi =\sum_{j \in \Z} \psi_j \in 
\bigoplus_{j \in \Z}P_{ac}(h_j)L^2(\R^+,\C^2) 
\end{align}
non-zero the estimate
\begin{align*}
\| |\bx|^p e^{-\ri tH}\psi \|^2 =
\sum_{j \in \Z}  \| r^p e^{-\ri t h_j}\psi_j \|^2  \ge
\sum_{j =-l}^{l} \| r^p e^{-\ri t h_j}\psi_j \|^2,
\end{align*}
where $l \in \N$ is chosen to be so large that 
$\sum_{j =-l}^{l}  \| \psi_j \|^2 \ge \tfrac{1}{2}\|\psi\|$.
Observing that 
$\mathbbm{1}_\Delta(H) = 
\bigoplus_{j \in \Z} \mathbbm{1}_\Delta(h_j)$, we 
deduce from Theorem \ref{lastmainthm2} and \eqref{kunigunde} that
\begin{align*}
\| |\bx|^p e^{-\ri tH}\psi \|^2 \ge 
\sum_{j =-l}^{l} C_j(\psi_j, \Delta ,p)\, T^p = C(\psi, \Delta ,p) \, T^p
\end{align*}
(where we set $C_j(\psi_j, \Delta ,p)=0$ if $\psi_j=0$), with 
$C(\psi, \Delta ,p) >0$.
\end{appendix}
\bibliographystyle{plain} 

\end{document}